\documentclass[12pt]{article}

\usepackage{amsthm}
\usepackage{amsmath} \usepackage{amsfonts}
\usepackage{amssymb} \usepackage{latexsym} \usepackage{enumerate}
\usepackage{mathtools}

 \usepackage{mathrsfs}
\usepackage[numbers]{natbib}

\input{style.tex}

\numberwithin{equation}{section}

\mathtoolsset{showonlyrefs}

\begin{document}

\title{Scaling Limits for Super--replication with Transient Price Impact}

\author{ Peter Bank\footnote{Technische Universit{\"a}t Berlin,
    Institut f{\"u}r Mathematik, Stra{\ss}e des 17. Juni 136, 10623
    Berlin, Germany, email \texttt{bank@math.tu-berlin.de}.  }
  \hspace{2ex} Yan Dolinsky\footnote{ Department of Statistics, Hebrew
    University and School of Mathematical Sciences, Monash University.
    email \texttt{yan.dolinsky@mail.huji.ac.il}.  The author YD is
    partially supported by the ISF Grant 160/17} }

\date{\today}

\maketitle
\begin{abstract}
  We prove a scaling limit theorem for the super-replication cost of
  options in a Cox--Ross--Rubinstein binomial model with transient
  price impact. The correct scaling turns out to keep the market depth
  parameter constant while resilience over fixed periods of time grows
  in inverse proportion with the duration between trading times. For
  vanilla options, the scaling limit is found to coincide with the one
  obtained by PDE-methods in \cite{GokaySoner:12} for models with
  purely temporary price impact. These models are a special case of
  our framework and so our probabilistic scaling limit argument allows
  one to expand the scope of the scaling limit result to
  path-dependent options.
 \end{abstract}

\begin{description}
\item[Mathematical Subject Classification (2010):] 91G10, 60F05

\item[Keywords:] Super-replication, Liquidity, Binomial model, Limit theorems
\end{description}

\maketitle

\markboth{P. Bank and Y. Dolinsky}{Scaling Limit with Transient Price Impact}
\renewcommand{\theequation}{\arabic{section}.\arabic{equation}}
\pagenumbering{arabic}

\section{Introduction}\label{sec:1}
 \label{intro}

 Super-replication in continuous-time financial models with market
 frictions is well known to typically lead to trivial buy-and-hold
 results. For markets with proportional transaction costs, this was
 first established rigorously by \cite{SSC95}. For
 discrete-time models, \cite{Kusuoka:95} used a dual description of
 super-replication costs to determine a regime that yields a
 non-trivial scaling limit for vanishing transaction costs. Such
 scaling limits were recently obtained in the multivariate case
 in \cite{BankDolinskyPerkkio:17}, for fixed costs in \cite{BankDolinsky:16}, and
 for purely temporary nonlinear costs as specified by
 \cite{CetinJarrProt:04} in \cite{GokaySoner:12,DolinskySoner:13,
   BankDolinskyGokay:16}.

The present paper yields such a scaling limit
 result for models with transient price impact where also past trades
 affect the spread at which present transactions are executed; see
 \cite{ObiWang:13, AlfFruSch:10, PredShaShr:11} for
 models of this type for optimal liquidation problems and
 \cite{BankVoss:18b} for an optimal investment study of such a
 model. The present paper is motivated by
 \cite{BankDolinsky:18} which confirms the triviality of
 super-replication costs also for continuous-time models with
 transient price impact. We therefore introduce in this paper a
 discrete-time version of the model considered in
 \cite{BankDolinsky:18} and, for the special case of a binomial
 Cox--Ross-Rubinstein reference model, compute the scaling limit of
 super-replication costs when market resilience becomes infinite.

 It turns out that the resulting scaling limit coincides with the
 scaling limit obtained for binomial models with purely temporary
 price impact and modified market depth, as studied (for the geometric
 random walk case) in \cite{GokaySoner:12,DolinskySoner:13,
   BankDolinskyGokay:16} . In this regard, it nicely complements the
 high-resilience asymptotics carried out by \cite{RochSoner:13} who
 prove convergence in probability of wealth dynamics.

 Our approach for computing the scaling limit is purely
 probabilistic. The proof of the lower bound is done in two steps.  In
 the first step, we establish a simple lower bound for the
 super--replication prices in terms of consistent price systems with
 ``small'' spread.  The second step is to use Kusuoka's techniques
 from \cite{Kusuoka:95} to construct, for a given martingale $M$ on
 Wiener space with suitably regular volatility process, a sequence of
 consistent price systems for our binomial reference models with
 vanishing spread which converge in law to $M$. Kusuoka's techniques
 are particularly useful here as they also allow us to control the
 approximation of the quadratic variation of $M$.

 The proof of the upper bound is more complicated. First, we notice
 that the portfolio value in the transient price impact dominates from
 above the portfolio value in a quadratic costs setup with a modified
 market depth which can be viewed as a binomial version of the
 temporary price impact model introduced in \cite{CetinJarrProt:04}.
 The key step is then to establish an upper bound for the
 super--replication prices with such quadratic costs. Using the
 pathwise Doob inequalities of \cite{Acciaioetal:13}, we argue that it
 essentially suffices to super--replicate the payoff knocked out when
 the underlying fluctuates ``too much''. For this ``tamed'' payoff, we
 identify a rich enough subclass of constrained trading strategies,
 for which super--replication costs remain unchanged asymptotically,
 but whose dual consistent price systems turn out to be tight. This
 new technique to obtain tightness in fully quadratic costs problems
 is key for our analysis and allows us to resolve an open question
 from \cite{DolinskySoner:13,BankDolinskyGokay:16} who had to impose
 linear growth constraints on transaction costs and only allowed
 quadratic costs in an ever smaller region around zero.  As a
 by-product of our probabilistic approach, we obtain an extension of
 the limit result of \cite{GokaySoner:12}, who used PDE-techniques, from
 vanilla options to path-dependent options.

  The paper is organized as follows. In Section~\ref{sec:2} we
  formalize the super-replication problem with transient price impact
  and give a duality result.  Section~\ref{sec:3} formulates and
  discusses our scaling limit result and gives its proof.

\section{Super-replication with transient price impact in discrete time}\label{sec:2}
\setcounter{equation}{0}
\subsection{A discrete-time model with transient price impact}

In this section we develop and analyze a discrete-time version of the
continuous-time financial model studied in \cite{BankDolinsky:18}
where the trades of a large investor affect an asset's price in a
transient manner. Specifically, we fix a filtered probability space
$(\Omega,\cF,(\cF_n)_{n=0,\dots,N},\P)$ and consider an adapted,
real-valued process $P=(P_n)_{n=0,\dots,N}$ to describe the evolution
of an asset's fundamental value at times $n=0,\dots,N$. In addition to
this asset, a large investor has at her disposal a bank account that,
for simplicity, bears no interest. She is endowed with an initial
position of $X_0 \set x_0\in\RR$ units of the asset and is free to
choose her position $X_n \in \cF_{n-1}$ in which she will confront the
$n$th fundamental shock $\Delta P_n \set P_n-P_{n-1}$,
$n=1,\dots,N$. We will let $\cX$ denote the collection of all these
strategies $X$. In line with \cite{HubermanStanzl:04}, the investor's
transactions have a linear permanent impact on the asset's price
beyond its fundamental value. So, the mid-price evolves according to
\begin{align*}
    P^X_n \set P_n + \iota X_n, \quad n=0,\dots,N.
\end{align*}
In addition, the investor's transactions affect the half-spread, i.e.,
the mark-up above (resp. below) the mid-price $P^X$ at which the
investor's orders are filled when she buys (respectively sells) the
asset. We model this quantity by
\begin{align}
    \zeta^X_0 &\set \zeta_0, \\ \label{eq:12a}
    \zeta^X_n &\set (1-r)\zeta^X_{n-1} +\frac{|X_n-X_{n-1}|}{\delta}, \quad n=1,\dots,N.
\end{align}
Here, $\zeta_0 \geq 0$ is the given initial half-spread. The
investor's trades widen the spread in inverse proportion to the
market's depth $\delta>0$, assumed to be constant for simplicity. The
constant $0 < r \leq 1$ measures the market's resilience and
describes the fraction by which the spread will diminish over a
trading period. It is convenient (and quite appropriate) to assume
that transactions affect mid-prices and spreads gradually, letting the
first bits of the $n$th transaction $X_n-X_{n-1}$ be filled at the
favorable pre-transaction mid-price $P_{n-1} + \iota X_{n-1}$ and at
the pre-transaction spread $(1-r)\zeta^X_{n-1}$ while the last bits
are filled at the less favorable post-transaction levels
$P_{n-1} + \iota X_n=P^X_n-\Delta P_n$ and
$(1-r)\zeta^X_{n-1}+|X_n-X_{n-1}|/\delta=\zeta^X_n$. As a result, the
investor's given cash position $\xi^X$ evolves from its given initial
level $\xi_0 \in \RR$ according to
\begin{align}
    \xi^X_0 & \set \xi_0,\\ \label{eq:10}
    \xi^X_n & \set \xi^X_{n-1}-\left(P_{n-1}+\frac{\iota}{2}(X_n+X_{n-1})\right)(X_n-X_{n-1})\\&\qquad -\left((1-r)\zeta^X_{n-1}+\frac{1}{2\delta}|X_n-X_{n-1}|\right)|X_n-X_{n-1}|
\end{align}
at times $n=1,\dots,N$. A more tangible description of the investor's
cash positions is given by the following lemma.
\begin{Lemma}\label{lem:wealthdynamics}
 The investor's cash position at time $n=1,\dots,N$ is
 \begin{align} \label{eq:5}
     \xi^X_n &= \xi_0 - \sum_{m=1}^n P_{m-1}(X_m-X_{m-1}) -
               \frac{\iota}{2}(X^2_n-x_0^2)-\kappa^X_n\\
&= \xi_0 +x_0P_0-X_nP_n+ \sum_{m=1}^n X_m (P_m-P_{m-1})\\
&\qquad -               \frac{\iota}{2}(X^2_n-x_0^2)-\kappa^X_n,
 \end{align}
 where $\kappa^X$ describes the liquidity costs
 \begin{align}\label{eq:11}
     \kappa^X_n &= (1-r) \sum_{m=1}^n \zeta^X_{m-1}|X_m-X_{m-1}|+\frac{1}{2\delta}\sum_{m=1}^n (X_m-X_{m-1})^2\\\label{eq:12}
     &= \frac{\delta}{2}\left((\zeta^X_n)^2+(1-(1-r)^2)\sum_{1 \leq m < n} (\zeta^X_m)^2-(1-r)^2\zeta_0^2\right)
 \end{align}
 with
  \begin{align}
     \zeta^X_n &= (1-r)^n\zeta_0 + \frac{1}{\delta}\sum_{m=1}^n (1-r)^{n-m}|X_m-X_{m-1}|, \quad n=1,\dots,N.
 \end{align}
\end{Lemma}
\begin{proof}
  Identity~\eqref{eq:5} follows readily from~\eqref{eq:10} where the
  representation~\eqref{eq:11} of $\kappa^X_n$ is due to
  $\sum_{m=1}^n (X_m+X_{m-1})(X_m-X_{m-1})=\sum_{m=1}^n
  X^2_m-X^2_{m-1}=X^2_n-x_0^2$; \eqref{eq:12} follows readily by
  expressing $|X_m-X_{m-1}|$ in terms of $\zeta^X_m$ and
  $\zeta^X_{m-1}$ as made possible by~\eqref{eq:12a}.
\end{proof}
In particular, we see that the liquidity costs $\kappa^X$ are a convex
functional of the investor's trading strategy $X \in \cX$. This
observation opens the door for convex duality methods that indeed will
be key for our subsequent analysis.

\subsection{Super-replication duality}

Having established the investor's wealth dynamics, we can now consider
the problem to super-replicate a contingent claim specified by a
payoff $H \in \cF_N$ unaffected by the investor's transactions. More
precisely, we will try to characterize the super-replication costs
\begin{align}
    \pi(H)\set \inf\{\xi_0 \;:\; \xi^X_N \geq H \text{ a.s. for some $X \in \cX$ with $X_N$=0}\}.
\end{align}
For models with full resilience ($r=1$) as in \cite{CetinJarrProt:04}, a dual description of
super-replication costs has been obtained in \cite{DolinskySoner:13}. For
models with limited resilience ($r \in (0,1)$) such a description is
given by the following lemma which complements its continuous-time
analogue established in \cite{BankDolinsky:18}:

\begin{Proposition}
 If $r \in [0,1)$, the super-replication costs of any contingent claim
 $H \geq 0$ have the dual description
 \begin{align}\label{eq:501}
     \pi(H) = \sup_{(\QQ,M,\alpha)} \left\{\E_{\QQ}\left[H\right]-\frac{1}{2}\E_{\QQ}\left[\sum_{n=1}^N|\alpha_n-\zeta_0|^2\mu_n\right] - M_0 x_0-\frac{\iota}{2} x_0^2\right\},
 \end{align}
 where
$$\mu_n \set \delta (1-(1-r)^2)(1-r)^{2n} \text{ for }
n=1,\dots,N-1, \text{ and } \ \mu_N=\delta (1-r)^{2N},$$ and where the
supremum is taken over all triples $(\QQ,M,\alpha) $ of measures
$\QQ\ll\P$, square-integrable $\QQ$-martingales $M$ and
$\QQ$-square-integrable, predictable processes $\alpha$ with
 \begin{align}
     |P_{n-1}-M_{n-1}| \leq \frac{1}{\delta(1-r)^n}\E_{\QQ}\left[\sum_{m=n}^N \alpha_m \mu_m\middle|\cF_{n-1}\right], \quad n=1,\dots,N.
 \end{align}
 In the case $r=1$ corresponding to purely temporary impact, we have the simpler
 duality
 \begin{align}\label{eq:5011}
     \pi(H) = \sup_{(\QQ,M)} \left\{\E_{\QQ}\left[H\right]-\frac{1}{2\delta}\E_{\QQ}\left[\sum_{n=1}^N|P_{n-1}-M_{n-1}|^2\right] - M_0 x_0-\frac{\iota}{2} x_0^2\right\}
 \end{align}
 with a supremum over probabilities $\QQ \ll \P$ and all square-integrable
 $\QQ$-martingales~$M$.
\end{Proposition}
\begin{proof}
Using the wealth dynamics of Lemma~\ref{lem:wealthdynamics}, the proof
can be done similarly as in the continuous-time analogue in
\cite{BankDolinsky:18} and is therefore omitted. For $r=1$ one can proceed as in
\cite{DolinskySoner:13} together with the Lagrange multiplier argument
for the choice of martingale $M$ from~\cite{BankDolinsky:18}.
\end{proof}
So, super-replication costs in our model with price impact take the
form of a convex risk measure. The structure of the costs' dual
description is similar in spirit to the one
observed for proportional transaction costs models: the payoff's assessment
is made using consistent price systems with a martingale $M$ that is
in some sense close to the underlying's price process $P$. By contrast
to these models with fixed spread, closeness is measured in our
setting by a process $\alpha$ that needs to be chosen to balance
greater flexibility in choosing $M$ with higher penalties from the
$L^2$-distance to the initial spread arising in the Legendre-Fenchel
representation~\eqref{eq:501} of the super-replication cost
functional.

As illustrated in \cite{BankDolinsky:18}, super-replication prices in
continuous-time often are trivially arising from simple buy-and-hold
strategies that cannot be improved upon due to the most unlikely, but
nonetheless still most relevant strong short-term fluctuations in the
price of the hedging instrument that are typically possible in these
models. Similar to Kusuoka's approach in \cite{Kusuoka:95} to discrete-time
models with fixed spread, we thus need to re-scale price impact to
ensure a non-trivial scaling limit for our model. This will be made
precise in the next section.

\section{Scaling limit of super-replication costs}\label{sec:3}

In this section we will derive a scaling limit result for the
super-replication costs from the previous section, letting the number
of trading periods $N$ over the time span $[0,1]$ tend to infinity
while re-scaling the time between trades as $1/N$. For the price
fluctuations we now focus on a binomial model where
$\Omega=\{-1,+1\}^{\{1,2,\dots\}}$ with coordinate maps
$\xi_n(\omega)=\omega_n$ indicating the upwards and downwards
movements of the fundamental asset value for scenario
$\omega=(\omega_n)_{n=1,2,\dots} \in \Omega$. The filtration
$(\cF_n)_{n=1,2,\dots}$ is generated by these coordinate maps and we
assume $\P$ to be the measure under which $\xi_1,\xi_2,\dots$ are
i.i.d.\ with $\P[\xi_n = -1] = \P[\xi_n=+1]=1/2$. Assuming an additive
model for the fundamental asset price we let, with the usual square
root scaling,
\begin{align}\label{eq:305}
    P^N_t \set p_0+\frac{\sigma}{\sqrt{N}}\sum_{n=1}^{[Nt]} \xi_n, \quad 0 \leq t \leq 1,
\end{align}
where $p_0 \in \RR$ is the initial fundamental asset price and $\sigma > 0$ the
asset's volatility. The price impact parameters $\iota \geq 0$, $r
\in (0,1]$ and $\delta>0$ are kept constant as we re-scale. As a result, the same
resilience effect is obtained over ever shorter time periods $1/N$,
implying a high-resilience limit in our scaling.

The main result of this paper is a scaling limit theorem for the
super-replication price of payoff profiles $h$ for which we will need
the following regularity assumption.

\begin{Assumption}\label{asp:hregular}
  The functional $h:D[0,1] \to \RR$ is nonnegative and Lipschitz
  continuous with respect to the Skorohod metric
 \begin{align}
   \label{eq:622}
   d(p,q) \set \inf_{\chi} \left\{\sup_{0 \leq t \leq 1}|t-\chi(t)|+\sup_{0
   \leq t \leq 1} |p(t)-q(\chi(t))|\right\}, \quad p,q \in D[0,1],
 \end{align}
 where the infimum is over all strictly increasing continuous time
 changes $\chi:[0,1] \to [0,1]$ with $\chi(0)=0$ and $\chi(1)=1$.
\end{Assumption}
Observe that the maps $p\rightarrow p(1)$,
$p\rightarrow \sup_{0\leq t\leq T}p(t)$ are Lipschitz continuous with
respect to the above Skorohod metric. Hence, call options, put options
and lookback options are covered in our setup; knock-out features,
however, will typically lead to discontinuities not covered by our
assumption.

This puts us in a position to state our limit theorem:

\begin{Theorem}\label{thm:1}
 For a payoff profile $h$ satisfying Assumption~\ref{asp:hregular},
 the super-replication costs $\pi^N(h(P^N))$ in the $N$-period model, $N=1,2,\dots$,
 have the high-resilience scaling limit
 \begin{align}
     \lim_N \pi^N(h(P^N)) = &\sup_{\nu\in \mathcal D}
                              \E_{\PP^W}\left[h(P^\nu)-\frac{r
                              \delta}{8 \sigma^2(2-r)}\int_0^1
                              |\nu_t^2-\sigma^2|^2\,dt\right] \\ &
                                                                   \qquad - P_0x_0-\frac{\iota}{2}x_0^2,
 \end{align}
 where $\mathcal D$ is the set of all bounded, nonnegative
 progressively measurable processes $\nu$ on the Wiener space
 $(\Omega^W,\cF^W,(\cF^W_t)_{0 \leq t \leq 1},\P^W)$ with Wiener
 process $W$ and where
 \begin{align}
     P^\nu_t \set p_0 + \int_0^t \nu_s \,dW_s, \quad 0 \leq t \leq 1.
 \end{align}
\end{Theorem}
The preceding theorem identifies the scaling limit of our
discrete-time super-replication prices in the form of a convex risk
measure. This measure assigns to a model, identified through its
volatility profile $\nu$, a penalty that is determined by its local
variances's $L^2$-distance from the reference variance
$\sigma^2$. Interestingly, this is also the scaling limit that emerges
from price impact models with purely temporary impact, albeit with a
different weight; see \cite{GokaySoner:12,DolinskySoner:13,
  BankDolinskyGokay:16}.

The connection between transient and
temporary impact for high-resilience limits has been observed before
in \cite{RochSoner:13} who prove convergence in probability for the
value processes. Our result complements this with a first rigorous
result in the context of super-replication.

On a technical level, it is worth mentioning that, to the best of our
knowledge, our proof below is the first purely probabilistic approach
which allows one to obtain a scaling limit result with fully quadratic
temporary costs. As a result, we are able to cover also sufficiently
regular, path-dependent options, thus extending beyond the vanilla
option case covered by the viscosity solution techniques of
\cite{GokaySoner:12}. The key challenge here is to find a setting
where one can prove tightness for a suitable sequence of dual
variables. This challenge is met by a judiciously chosen set of
constrained hedging strategies in our proof of the upper bound.

\subsection{Proof of the lower bound}

In this section we will prove
\begin{align}
     \liminf_N \pi^N(h(P^N)) &\geq \sup_{\nu \in \cD} \E_{\PP^W}\left[h(P^\nu)-\frac{r \delta}{8 \sigma^2(2-r)}\int_0^1 |\nu_t^2-\sigma^2|^2\,dt\right]\\ &
                                                                   \qquad - P_0x_0-\frac{\iota}{2}x_0^2.\label{eq:4d}
 \end{align}
 Let us start by observing that, by the density arguments of Lemma~7.3
 in \cite{DolinskySoner:13}, the above supremum coincides with the one
 taken over the class $\cD_0$ of volatility profiles $\nu \in \cD$
 which are bounded away from zero and Lipschitz in the sense that for
 some constant $C>0$ we have
\begin{align*}
    \nu_t(\omega) &\geq 1/C, \\
    |\nu_t(\omega)-\nu_{t'}(\omega')| &\leq C\left(|t-t'|+\sup_{s \in [0,1]} |\omega(s)-\omega'(s)|\right)
\end{align*}
for $t,t' \in [0,1]$, $\omega,\omega'\in C[0,1].$ For any such $\nu$,
the seminal paper \cite{Kusuoka:95} constructs probabilities with martingales
``close'' to the random walk $P^N$ which in distribution converge to
$P^\nu=p_0 + \int_0^. \nu_s dW_s$ as summarized in the following
lemma.
\begin{Lemma}\label{lem:Kusuoka}
For any $\nu \in \cD_0$, there is a sequence of probability measures
$\QQ^N$ on $(\Omega,\cF_N)$ and $(\cF_n)_{n=0,\dots,N}$-predictable
processes $\alpha^N=(\alpha^N_n)_{n=1,\dots,N}$, $N=1,2,\dots$, such
that for some constant $C>0$ independent of $N$ we have
\begin{enumerate}
    \item $|\alpha^N_n| \leq C$, $|\alpha^N_n-\alpha^N_{n-1}|\leq C/\sqrt{N}$,  $n=1,\dots,N$;
    \item $M^N_0 \set P_0$, $M^N_n \set P^N_{n/N}+\alpha^N_n\xi_n/\sqrt{N}$, $n=1,\dots,N$, is a $\QQ^N$-martingale;
    \item $\Law\left((P^N_{[Nt]},\alpha^N_{[Nt]})_{0 \leq t \leq 1} \middle| \QQ^N\right) \to
    \Law\left((P^\nu_t, (\nu_t^2-\sigma^2)/({2\sigma}))_{0 \leq t \leq 1} \middle| \PP^W\right)$ weakly on $D[0,1]$ as $N \uparrow \infty$.
\end{enumerate}
\end{Lemma}
\begin{proof}
  Adjusting for the additive setting considered here, this follows
  exactly as in Kusuoka's original approach for the multiplicative
  geometric random walk setting from \cite{Kusuoka:95}.
\end{proof}

With the above approximation result and the representation of
liquidity costs~\eqref{eq:11}, \eqref{eq:12} at hand, we are now in a
position to prove~\eqref{eq:4d}. Indeed, take $\QQ^N$ and $\alpha^N$
as in the preceding lemma and observe that, for any $N$-period
strategy $X=(X_n)_{n=0,\dots,N}$ with $X_N=0$, we can estimate
\begin{align*}
 -\E_{\QQ^N}&\left[\sum_{m=1}^N P^N_{(m-1)/N}(X_m-X_{m-1})\right]
 \\& = -\E_{\QQ^N}\left[\sum_{m=1}^N (P^N_{(m-1)/N}-M^N_N)(X_m-X_{m-1})+M^N_N(X_N-x_0)\right]
 \\&=-\E_{\QQ^N}\left[\sum_{m=1}^N (P^N_{(m-1)/N}-M^N_{m-1})(X_m-X_{m-1})\right]+M^N_0x_0
 \\&\leq \E_{\QQ^N}\left[\sum_{m=1}^N \frac{|\alpha^N_{m-1}|}{\sqrt{N}}|X_m-X_{m-1}|\right]+M^N_0x_0
 \\&= \E_{\QQ^N}\left[\sum_{m=1}^N \frac{|\alpha^N_{m-1}|}{\sqrt{N}}\delta\left(\zeta^X_m-(1-r)\zeta^X_{m-1}\right)\right]+P_0x_0
\end{align*}
where we used the martingale property of $M^N$ along with $X_N=0$ for the second identity and the second property of $\alpha^N$ listed in Lemma~\ref{lem:Kusuoka} for the estimate. Hence, for $X \in \cX$ with $X_N=0$ which super-replicates $h(P^N)$ in the sense that $\xi^X_N\geq h(P^N)$ we can estimate
\begin{align*}
    \E_{\QQ^N}&\left[h(P^N)\right] \leq \E_{\QQ^N}\left[\xi^X_N\right]
    \\&=\xi_0-\E_{\QQ^N}\left[\sum_{m=1}^N P^N_{(m-1)/N}(X_m-X_{m-1})+\frac{\iota}{2}(X^2_N-x_0^2)+\kappa^X_N\right]
    \\&\leq \xi_0+P_0x_0+\frac{\iota}{2}x_0^2+
    \\&\qquad +
       \delta\frac{\alpha^N_{0}}{\sqrt{N}}\zeta_{0}-\frac{1}{2}(\zeta_0)^2
    \\&\qquad + \delta\E_{\QQ^N}\left[\sum_{1 \leq
        m<N}\left(\frac{|\alpha^N_{m-1}|}{\sqrt{N}}-(1-r)\frac{|\alpha_m|}{\sqrt{N}}\right)\zeta^X_{m}-\frac{1-(1-r)^2}{2}(\zeta^X_m)^2\right]
   \\&\qquad +\delta \E_{\QQ^N}\left[\frac{\alpha^N_{N-1}}{\sqrt{N}}\zeta^X_N-\frac{1}{2}(\zeta^X_N)^2\right].
\end{align*}
Using the estimate $a \zeta-\frac{c}{2}\zeta^2 \leq a^2/(2c)$ in each of the
last three lines and rearranging terms yields
\begin{align*}
 \xi_0&+P_0x_0+\frac{\iota}{2}x_0^2+\frac{\delta
                              (\alpha^N_0)^2}{2N}+\delta
                              \E_{\QQ^N}\left[\frac{(\alpha^N_{N-1})^2}{2N}\right]
  \\
&\geq \E_{\QQ^N}\left[h(P^N)\right]-  \delta\E_{\QQ^N}\left[\sum_{1 \leq
  m<N}\frac{(|\alpha^N_{m-1}|-(1-r)|\alpha_m|)^2}{2(1-(1-r)^2)N}\right]\\
& \geq \E_{\QQ^N}\left[h(P^N)\right]- \frac{\delta }{2(1-(1-r)^2)}\E_{\QQ^N}\left[\frac{1}{N}\sum_{1 \leq
  m<N}(r|\alpha^N_{m-1}|+C/\sqrt{N})^2\right]
\end{align*}
where in the last estimate we used the first property of $\alpha^N$
from Lemma~\ref{lem:Kusuoka}. The same property also yields the
uniform boundedness of $\alpha^N$, $N=1,2,\dots$, and so the third
property listed in Lemma~\ref{lem:Kusuoka} in conjunction with the
regularity assumption~\ref{asp:hregular} on $h$ thus allows us to pass
to the limit $N \uparrow \infty$ in the above estimate to conclude
that
\begin{align*}
  \liminf_N &\pi^N(h(P^N))+P_0x_0+\frac{\iota}{2}x_0^2 \\&\geq \E_{\P^W} \left[h(P^\nu)\right]-
                                      \frac{\delta }{2(1-(1-r)^2)}
                                      \E_{\P^W} \left[\int_0^1
                                      \left(r\frac{\nu_s^2-\sigma^2}{2\sigma}\right)^2ds\right]
  \\
& = \E_{\P^W} \left[h(P^\nu)-
                                      \frac{\delta r }{8(2-r)\sigma^2}
                                      \int_0^1
                                      \left(\nu_s^2-\sigma^2\right)^2ds\right].
\end{align*}
This yields the desired lower bound~\eqref{eq:4d}.

\subsection{Proof of the upper bound}
We will prove the upper bound first for the case $$x_0=\zeta_0=0$$ and
reduce the general case to this one in the end.

For the upper bound
\begin{align}
     \limsup_N \pi^N(h(P^N)) \leq
     &\sup_{\nu \in \cD} \E_{\PP^W}\left[h(P^\nu)-\frac{r \delta}{8 \sigma^2(2-r)}\int_0^1 |\nu_t^2-\sigma^2|^2\,dt\right]\label{eq:4}
\end{align}
we first note that super-replication prices with transient impact are dominated by super-replication prices in a suitable model with purely temporary impact as in \cite{CetinJarrProt:04}:
\begin{Lemma}
 For any $N=1,2,\dots$, we have
 \begin{align}
     \pi^N(h(P^N)) \leq \widehat{\pi}^N(h(P^N))
 \end{align}
 where $\widehat{\pi}^N(h(P^N))$ is the super-replication price in the
 model with full resilience $\widehat{r} \set 1$ and market depth $\widehat{\delta}\set r\delta/(2-r)$.
\end{Lemma}
\begin{proof}
 Consider the cost term $\kappa^X_n$ from~\eqref{eq:11} and observe that with $\zeta_0=0$,
 \begin{align}
 (1-r)& \sum_{m=1}^n \zeta^X_{m-1}|X_m-X_{m-1}|    \\
 &= \frac{1-r}{\delta} \sum_{m=1}^n \sum_{1 \leq l<m} (1-r)^{m-1-l}|X_l-X_{l-1}||X_m-X_{m-1}|\\
 &\leq \frac{1-r}{\delta} \sum_{m=1}^n \sum_{1 \leq l<m} (1-r)^{m-1-l}\frac{1}{2}(|X_l-X_{l-1}|^2+|X_m-X_{m-1}|^2)\\
 &= \frac{1-r}{\delta} \sum_{m=1}^n \frac{1}{2}\left(\sum_{l=1}^{m-1}(1-r)^{m-1-l}+\sum_{l=m}^{n-1}(1-r)^{l-m}\right)|X_m-X_{m-1}|^2\\
 &\leq \frac{1-r}{\delta}  \sum_{m=1}^n \left(\sum_{k=0}^\infty(1-r)^{k}\right)|X_m-X_{m-1}|^2\\
 &= \frac{1-r}{\delta r} \sum_{m=1}^n|X_m-X_{m-1}|^2.
 \end{align}
 As a result the cost term $\kappa^X_n$ in the original model can be estimated by
 \begin{align}
     \kappa^X_n \leq
     \left(\frac{1-r}{\delta r}+\frac{1}{2\delta}\right) \sum_{m=1}^n|X_m-X_{m-1}|^2 = \widehat{\kappa}^X_n
 \end{align}
 where $\widehat{\kappa}^X$ is the cost term for the fully resilient model with $\widehat{r}=1$ and depth $\widehat{\delta}=r\delta/(2-r)$. The costs in this auxiliary model being higher, the super-replication of any claim cannot be less expensive than in the original model and we obtain our assertion.
\end{proof}

For~\eqref{eq:4} it thus suffices to prove
\begin{align}
     \limsup_N \widehat{\pi}^N(h(P^N)) \leq \sup_{\nu \in \cD} \E_{\PP^W}\left[h(P^\nu)-\frac{\widehat{\delta}}{8 \sigma^2}\int_0^1 |\nu_t^2-\sigma^2|^2\,dt\right].\label{eq:14}
\end{align}
 For this asymptotic analysis, we will work with a family of space-time discretizations of our price process. Specifically, we let, for any $\epsilon>0$, the sequence of partitions $\tau^{N,\epsilon}=(\tau^{N,\epsilon}_k)_{k=0,1,\dots}$, $N=1,2,\dots$, be given by $\tau^{N,\epsilon}_0 \set 0$ and
\begin{align*}
    \tau^{N,\epsilon}_k \set \inf\left\{t \geq \tau^{N,\epsilon}_{k-1}\;:\; |P^N_t-P^N_{\tau^{N,\epsilon}_{k-1}}| \geq \epsilon \text{ or } |t-\tau^{N,\epsilon}_{k-1}| \geq \epsilon^2\right\} \wedge (1-N^{-2/3})
\end{align*}
for $k=1,2,\dots$. With $\tau^{N,\epsilon}$ we associate the  following discretization of $P^N$:
\begin{align*}
    P^{N,\epsilon}_t \set \sum_{k = 1,2,\dots} P^N_{\tau^{N,\epsilon}_{k-1}} 1_{[\tau^{N,\epsilon}_{k-1},\tau^{N,\epsilon}_k)}(t)+P^N_{1-N^{-2/3}}1_{[1-N^{-2/3},1]}(t), \quad 0 \leq t \leq 1.
\end{align*}

Our next lemma reveals that, under our regularity assumptions on $h$, the super-replication price of $h(P^N)$ in the $N$-step model is controlled by the super-replication price of a particular quadratic claim on $P^{N,\epsilon}$ and a knock-out variant of the claim $h$ applied to $P^{N,\epsilon}$ that only generates a payoff if this underlying does not fluctuate ``too much'':

\begin{Lemma}\label{lem:epsilondiscretization}
  Let $c=c(\lambda)>0$ be such that $h(p) \leq \lambda^2(\|p-p_0\|^2_{\infty}+c)$,
  $p \in D[0,1]$ (observe that $c(\lambda)$ exists since Assumption
  \ref{asp:hregular} implies that $h(p)$ has at most linear growth in $||p-p_0||_{\infty}$).
  Then, for any $\epsilon>0$ and $\lambda \in (0,1)$,
  the constant $K =K(\epsilon,\lambda) \set [c/(\epsilon\lambda)^2]+1$ is large enough to ensure that
  for sufficiently large $N$ we have
\begin{align}\label{eq:6}
    \widehat{\pi}^N(h(P^N)) \leq 3L\epsilon+(1-\lambda) \widehat{\pi}^N(H^{N,\epsilon,K}/(1-\lambda))
    +\lambda \widehat{\pi}^N(\lambda Q^{N,\epsilon})
\end{align}
where $L$ is the Lipschitz constant from Assumption \ref{asp:hregular}
and where
\begin{align}\label{eq:7}
    H^{N,\epsilon,K} &\set h(P^{N,\epsilon}) 1_{\{\tau^{N,\epsilon}_K=1-N^{-2/3}\}},\\
    Q^{N,\epsilon} &\set \sup_{0 \leq t \leq 1}|P^{N,\epsilon}_t-P_0|^2
    +\sum_{k=1,2,\dots}\left(|P^N_{\tau^{N,\epsilon}_k}-P^N_{\tau^{N,\epsilon}_{k-1}}|^2+|\tau^{N,\epsilon}_k-\tau^{N,\epsilon}_{k-1}|\right).
\end{align}
\end{Lemma}
\begin{proof}
 From the definitions of $L$ and $P^{N,\epsilon}$ and the regularity
 of $h$, it follows that, for sufficiently large $N$,
 \begin{align}
   \label{eq:300}
    h(P^N) \leq 3L\epsilon + h(P^{N,\epsilon}).
 \end{align}
 For $K=K(\epsilon,\lambda)$ as defined above, we have furthermore
 \begin{align}
   \label{eq:301}
   h(p) \leq \lambda^2\left(\sup_{0 \leq t \leq 1}
   |p(t)-p_0|^2+K\epsilon^2\right), \quad p \in D[0,1].
 \end{align}
 From the definition of $\tau^{N,\epsilon}_k$, $k=0,1,\dots$,  we get
 in addition that
 \begin{align}
   \label{eq:302}
   K \epsilon^2 \leq \sum_{k=1}^K
   \left(|P^N_{\tau^{N,\epsilon}_k}-P^N_{\tau^{N,\epsilon}_{k-1}}|^2+|\tau^{N,\epsilon}_k-\tau^{N,\epsilon}_{k-1}|\right)
   \text{ on } \left\{\tau^{N,\epsilon}_K<1-N^{-\frac{2}{3}}\right\}.
 \end{align}
Combining~\eqref{eq:301} and \eqref{eq:302} gives
\begin{align}
  \label{eq:303}
   h(P^{N,\epsilon}) \leq H^{N,\epsilon,K}+ \lambda^2 Q^{N,\epsilon} .
\end{align}
 Convexity of the wealth dynamics~\eqref{eq:5} implies convexity of the
 super-replication cost functional, and so~\eqref{eq:303} yields
 \begin{align}
   \widehat{\pi}^N(h(P^{N,\epsilon})) \leq (1-\lambda)
   \widehat{\pi}^N(H^{N,\epsilon,K}/(1-\lambda))+\lambda
   \widehat{\pi}^N(\lambda Q^{N,\epsilon}).
 \end{align}
Together with~\eqref{eq:300}, this implies~\eqref{eq:6}.
\end{proof}

Our next lemma shows that the super-replication price of $\lambda Q^{N,\epsilon}$ is easy to control (at least for small $\lambda \in (0,1)$) and so its contribution to~\eqref{eq:6} vanishes as $\lambda \downarrow 0$:

\begin{Lemma}\label{lem:quadraticclaim}
There exists $\lambda_0>0$ such that
for any $\epsilon>0$, $\lambda \in [0,\lambda_0]$ and $N=1,2,\dots$ we have
$$\widehat{\pi}^N(\lambda Q^{N,\epsilon})\leq \lambda(1+36\sigma^2).$$
\end{Lemma}
\begin{proof}
Let $a,b,d,e>0$ and consider a portfolio strategy with initial capital
$\xi_0=a$ and a (predictable) trading strategy of the form which for
$n$ from $[N\tau^{N,\epsilon}_{k-1}]+1$ to $[N\tau^{N,\epsilon}_{k}]$,
$k=1,2,\dots$, is given by
\begin{align*}
X_n=&-b\max_{i=0,\dots, k-1}\left(P^{N}_{\tau^{N,\epsilon}_{i}}-p_0\right)
+b\max_{0\leq i\leq k-1}\left(p_0-P^{N}_{\tau^{N,\epsilon}_{i}}\right)\\
&-d \left(P^{N}_{\tau^{N,\epsilon}_{k-1}}-p_0\right)+e\left(P^N_{(n-1)/N}-p_0\right),
\end{align*}
and which is 0 for $n$ from $[N(1-N^{-2/3})]+1$ to $N$.  In order to
estimate the corresponding portfolio value at the maturity date we
apply Proposition 2.1 in \cite{Acciaioetal:13} for $p=2$. We notice
that for $p=2$ this proposition holds true for any sequence of real
numbers, including negative numbers.  We apply this pathwise Doob's inequality for
$$s_k \set \pm (P^{N}_{\tau^{N,\epsilon}_{k}}-p_0), \quad k=1,2,\dots .$$
Moreover, we will use the elementary identity
\begin{eqnarray*}
 \sum_{k=1}^{j} y_k(y_{k+1}-y_k) & =\frac{1}{2}\left(y^2_{j+1}-y^2_1-\sum_{k=1}^{j}(y_{k+1}-y_k)^2\right)
 \end{eqnarray*}
with $y_k=P^{N}_{\tau^{N,\epsilon}_{k-1}}-p_0$ and also with $y_k=P^{N}_{(k-1)/N}-p_0$.
 By the well-known inequalities
 $$(z_1+z_2)^2\leq 2(z^2_1+z^2_2), \ \ (z_1+z_2+z_3+z_4)^2\leq 4(z^2_1+z^2_2+z^2_3+z^2_4),$$
 the result then is
 \begin{align*}
\xi^X_N=&a+\sum_{n=1}^N X_n(P^N_{n/N}-P^N_{(n-1)/N)})-\frac{1}{2\widehat\delta}\sum_{n=1}^N|X_n-X_{n-1}|^2\\
\geq& a+\frac{b}{4}\left(\max_{0 \leq t \leq 1}P^{N,\epsilon}_t-p_0\right)^2-b|P^{N,\epsilon}_1-p_0|^2\\
&+\frac{b}{4}\left(p_0-\min_{0 \leq t \leq 1}P^{N,\epsilon}_t\right)^2-b|P^{N,\epsilon}_1-p_0|^2\\
&+\frac{d}{2} \sum_{k=1,2,\dots}|P^N_{\tau^{N,\epsilon}_k}-P^N_{\tau^{N,\epsilon}_{k-1}}|^2-\frac{d}{2}|P^{N,\epsilon}_1-p_0|^2\\
&+\frac{e}{2}|P^{N,\epsilon}_1-P_0|^2-\frac{e}{2}\sigma^2\\
&-\frac{2}{\widehat\delta} \left(e^2\sigma^2+(2b^2+d^2)\sum_{k=1,2,\dots} |P^N_{\tau^{N,\epsilon}_k}-P^N_{\tau^{N,\epsilon}_{k-1}}|^2\right)\\
&-\frac{1}{2\widehat\delta}(2b+d+e)^2 \left(\frac{\sigma}{\sqrt N}+\max_{0 \leq t \leq 1}|P^{N,\epsilon}_t-p_0|\right)^2.
\end{align*}
Here, the last two lines give an estimate for the transaction costs
(including the liquidation costs) which correspond to our trading
strategy.

It follows that
\begin{align*}
\xi^X_N\geq& a-\sigma^2\left(\frac{e}{2}+\frac{2e^2}{\widehat\delta}+\frac{(2b+d+e)^2}{\widehat\delta}\right)\\
&+\left(\frac{b}{4}-\frac{(2b+d+e)^2}{\widehat\delta}\right)\sup_{0 \leq t \leq 1}|P^{N,\epsilon}_t-p_0|^2\\
&+\frac{e-4b-d}{2}|P^{N,\epsilon}_1-p_0|^2\\
&+\left(\frac{d}{2}-\frac{4b^2+2d^2}{\widehat\delta} \right)\sum_{k=1,2,\dots} |P^N_{\tau^{N,\epsilon}_k}-P^N_{\tau^{N,\epsilon}_{k-1}}|^2.
\end{align*}
Let $b=8\lambda$, $d=4\lambda$, $e=4b+d=36\lambda$ and
$a=\lambda+e\sigma^2=\lambda(1+36\sigma^2)$. Then for sufficiently small $\lambda$ we get
$$\xi^X_N\geq \lambda+ \lambda \sup_{0 \leq t \leq 1}|P^{N,\epsilon}_t-P_0|^2
    +\lambda\sum_{k=1,2,\dots} |P^N_{\tau^{N,\epsilon}_k}-P^N_{\tau^{N,\epsilon}_{k-1}}|^2\geq \lambda Q^{N,\epsilon}$$
and the result follows.
\end{proof}

The proof of the upper bound thus relies on an understanding how to super-replicate the claims $H^{N,\epsilon,K}/(1-\lambda)$. Notice that these claims depend on the values of their underlying at only a fixed number $K$ of sampling times. Such claims turn out to allow for a particularly convenient duality estimate for their super-replication prices:

\begin{Lemma}\label{lem:dualityestimate}
  Let $G$ be a claim of the form
  $G=g\left(\left(\tau^{N,\epsilon}_k,P^N_{\tau^{N,\epsilon}_k}\right)_{k=0,\dots,K}\right)$
  for some measurable,  bounded, nonnegative function
  $g=g((t_k,p_k)_{k=0,\dots,K})$. Then, for any
  $\epsilon,  \eta>0$, we can find for sufficiently large $N$ a
  probability $\QQ^N$ on $(\Omega,\cF_N)$ (also depending on
  $\epsilon$, $\eta$ and $g$) such that for the filtration
  $(\cF^{N,\epsilon}_k)_{k=0,\dots,K}$ generated by
  $(\tau^{N,\epsilon}_k,P^{N,\epsilon}_{\tau^{N,\epsilon}_k})_{k=0,\dots,K}$
  we have
\begin{align}
    &\widehat{\pi}^N\left(G\right) \leq \;\frac{1}{4}\eta \sigma^2 \widehat{\delta}+\E_{\QQ^N}\left[G\right]\label{eq:5new}\\
    & - \frac{\widehat{\delta}}{8 \sigma^2}
  \E_{\QQ^N}\left[\sum_{k=1}^K\left(\frac{\E_{\QQ^N}\left[(P^N_{\tau^{N,\epsilon}_k})^2-(P^N_{\tau^{N,\epsilon}_{k-1}})^2\;\middle|\;\cF^{N,\epsilon}_{k-1}\right]}
  {\E_{\QQ^N}\left[\eta+\tau^{N,\epsilon}_k-\tau^{N,\epsilon}_{k-1}\;\middle|\;\cF^{N,\epsilon}_{k-1}\right]}-\sigma^2\right)^2(\eta+\tau^{N,\epsilon}_k-\tau^{N,\epsilon}_{k-1})\right].
\end{align}
 In addition, under $\QQ^N$, $(P^N_{\tau^{N,\epsilon}_k})_{k=0,\dots,K}$ is close to being a
 martingale in the sense that
 \begin{align}
   \label{eq:401}
  \E_{\QQ^N}\left[\sum_{k=1}^K\left|\E_{\QQ^N}\left[P^N_{\tau^{N,\epsilon}_k}-P^N_{\tau^{N,\epsilon}_{k-1}}\middle|\cF^{N,\epsilon}_{k-1}\right]\right|\right]
\leq (\|G\|_\infty+\eta)/\log N.
 \end{align}
\end{Lemma}
\begin{proof}
  Fix $\epsilon>0$ and $N \in \{1,2,\dots\}$. Rather than looking
  among all strategies in $\cX$ for a cost-effective super-hedge, we
  will consider a suitably constrained class. To this end, denote by
  $\cA$ the class of pairs $(\phi_k,\psi_k)_{k=1,\dots,K}$ of
  $(\cF^{N,\epsilon}_k)_{k=0,\dots,K}$-predictable processes such that
  $|\phi_k|,|\psi_k| \leq \log N$ for $k=1,\dots,K$. Each such pair
  induces a strategy $(X^{(\phi,\psi)}_n)_{n=1,\dots,N} \in \cX$ in
  the $N$-step model that we can define piecewise on
  $\left[[N \tau^{N,\epsilon}_{k-1}], [N
    \tau^{N,\epsilon}_{k}]\right]$ for $k=1,\dots,K$ as follows: If
  $\tau^{N,\epsilon}_{k-1}<1-N^{-1/2}$, the duration
  $\tau^{N,\epsilon}_{k} -\tau^{N,\epsilon}_{k-1}$ of the $k$th period
  is at least of order $N^{1/2}$ and we thus can subdivide the
  interval
  $\left[[N \tau^{N,\epsilon}_{k-1}], [N
    \tau^{N,\epsilon}_{k}]\right)$ into two parts. On the first
  (short) subinterval of length $N^{1/3}$ we trade at constant speed
  into a position holding
  $\phi_k+\psi_k P^N_{\tau^{N,\epsilon}_{k-1}}$ risky assets; in the
  periods $n$ afterwards, we hold the position $\phi_k+\psi_k P^N_{n}$
  until the stopping time $[N\tau^{N,\epsilon}_{k}]+1$ when the next
  iteration of this recipe proceeds with $k+1$ instead of $k$ while
  $k<K$. If $\tau^{N,\epsilon}_{k-1} \geq 1-N^{-1/2}$ or when we have
  completed the $K$th such iteration, we complete the construction of
  the strategy by liquidating the obtained position in $N^{1/3}$ steps
  and staying flat until the end.

  Let us analyze the profits and losses and also the costs accruing
  from this strategy. For this, note that, due to the random walk
  dynamics~\eqref{eq:305}, we have
  \begin{align}\label{eq:306}
    \sum_{l < m \leq n} P^N_{\frac{m-1}{N}}(P^N_{\frac{m}{N}}-P^N_{\frac{m-1}{N}})
= \frac{1}{2}\left((P^N_{\frac{n}{N}})^2-(P^N_{\frac{l}{N}})^2-\sigma^2\frac{n-l}{N}\right).
  \end{align}
  Moreover, note that $P^N$ is uniformly bounded on
  $[0,\tau^{N,\epsilon}_K]$ by $|p_0|+K {\epsilon}+\sigma$, so that, in
  particular, $X^{(\phi,\psi)}$ is of size $O(\log N)$.
  Therefore, the profit and loss due to fluctuations in the fundamental
  value incurred by the above strategy is up to a term of order
  $O(\log N/N^{1/6})$ (accounting for the transition period of length
  $N^{1/3}$ when a position change of at most order $\log N$ is
  accomplished while the underlying moves in steps of order $1/\sqrt{N}$) given by
  \begin{align}\label{eq:104}
   \sum_{k=1}^K &
      \phi_k\left(P^N_{\tau^{N,\epsilon}_k}-P^N_{\tau^{N,\epsilon}_{k-1}}\right)\\
     &+ \sum_{k=1}^K \psi_k \frac{1}{2} \left((P^N_{\tau^{N,\epsilon}_k})^2-(P^N_{\tau^{N,\epsilon}_{k-1}})^2-\sigma^2(\tau^{N,\epsilon}_k-\tau^{N,\epsilon}_{k-1})\right).
  \end{align}
  At the same time, the logarithmic bounds on the allowed positions
  ensure that the costs are
  \begin{align}\label{eq:1033}
    \widehat{\kappa}^{X^{(\phi,\psi)}}_N = \frac{1}{2\widehat{\delta}}
    \sum_{k=1}^K \sigma^2
    \psi_k^2(\tau^{N,\epsilon}_k-\tau^{N,\epsilon}_{k-1}) + O(\log^2 N/N^{1/3}).
  \end{align}
  where the $O$-term accounts for the
  $O(N^{1/3}(\log N/N^{1/3})^2)=O(\log^2 N/N^{1/3})$ costs for the
  gradual position build-up of the first $N^{1/3}$ steps of each
  period $k=1,\dots,K$ and for the $O(\log^2 N/N^{1/3})$ costs
  resulting from the one possible jump at the end of this initial build-up which
  is at most of size $O(N^{1/3} \log N /\sqrt{N})=O(\log N/N^{1/6})$;
  the running costs in the second leg of each trading period
  $k=1,\dots,K$ are reflected by the sum in~\eqref{eq:1033}.

  It follows that, for large enough $N$, we will have
  \begin{align}\label{eq:105}
   \widehat{\pi}^N\left(G\right)
    \leq o(1)+\widetilde{\pi}^N\left(G\right)
  \end{align}
  where $\widetilde{\pi}^N\left(G\right)$ denotes the super-replication price of
  the claim $G$
  when restricting to strategies $X^{(\phi,\psi)}$ as above with
 profits and losses given by~\eqref{eq:104} and trading costs given by the sum in ~\eqref{eq:1033}.

 Observing that~\eqref{eq:104} is linear in $\phi$ and recalling the
 constraint $|\phi|\leq \log N$, we get from
 classical linear super-replication duality with convexly constrained
 strategy sets (cf. \cite{FollKram:97}, Theorem~4.1 in
 connection with Example~2.3) that for any fixed $\psi$-component we
 have
\begin{align}\label{eq:103}
   \widetilde{\pi}^N (G)\leq \sup_{\QQ}
  \E_{\QQ}\left[G^{\psi}-\log N
  \sum_{k=1}^K\left|\E_{\QQ}\left[P^N_{\tau^{N,\epsilon}_k}-P^N_{\tau^{N,\epsilon}_{k-1}}\middle|\cF^{N,\epsilon}_{k-1}\right]\right|\right]
 \end{align}
 where the supremum is taken over the set of all measures $\QQ$ on $(\Omega,\cF_N)$
 and where
 \begin{align}
  G^{\psi} \set G&-\sum_{k=1}^K \psi_k
                                        \frac{1}{2}
                                        \left((P^N_{\tau^{N,\epsilon}_k})^2-(P^N_{\tau^{N,\epsilon}_{k-1}})^2-\sigma^2(\tau^{N,\epsilon}_k-\tau^{N,\epsilon}_{k-1})\right)\\
&+\frac{1}{2\widehat{\delta}}
    \sum_{k=1}^K \sigma^2
    \psi_k^2(\tau^{N,\epsilon}_k-\tau^{N,\epsilon}_{k-1})
 \end{align}
 denotes the claim that remains to be super-hedged by suitably choosing $\phi$
 when the $\psi$-component is fixed.

 For $\cE(\QQ,\psi)$ denoting the unconditional expectation
 in~\eqref{eq:103}, it is readily checked that
 $\psi \mapsto \cE(\QQ,\psi)$ is convex for any fixed $\QQ$ and that
 $\QQ \mapsto \cE(\QQ,\psi)$ is concave for any $\psi$
 fixed. Observing that the domains of $\QQ$ and $\psi$ can easily be
 identified with convex and compact subsets in Euclidean space, we can
 thus invoke the Minimax Theorem (e.g. Theorem~45.8 in
 \cite{Strasser:85}) to obtain
 \begin{align}
    \widetilde{\pi}^N\left(G\right) \leq \inf_{\psi} \sup_{\QQ} \cE(\QQ,\psi) = \sup_{\QQ} \inf_{\psi} \cE(\QQ,\psi).
 \end{align}
 In conjunction with~\eqref{eq:105}, we therefore can find a $\QQ^N$
 on $(\Omega,\cF_N)$ such that
 \begin{align}\label{eq:110}
   \widehat{\pi}^N\left(G\right) \leq o(1)+\inf_\psi \cE(\QQ^N,\psi).
 \end{align}
 In order to control the latter infimum, observe that the terms in
 $\cE(\QQ^N,\psi)$ involving $\psi_k$ contribute
 \begin{align}\
   \E_{\QQ^N}&\left[\frac{\sigma^2}{2\widehat{\delta}}
        \psi_k^2(\tau^{N,\epsilon}_k-\tau^{N,\epsilon}_{k-1})- \psi_k
                                        \frac{1}{2}
                                        \left((P^N_{\tau^{N,\epsilon}_k})^2-(P^N_{\tau^{N,\epsilon}_{k-1}})^2-\sigma^2(\tau^{N,\epsilon}_k-\tau^{N,\epsilon}_{k-1})\right)\right]\\
 & \leq \E_{\QQ^N}\Bigg[\frac{\sigma^2}{2\widehat{\delta}}
        \psi_k^2\E_{\QQ^N}\left[\eta+\tau^{N,\epsilon}_k-\tau^{N,\epsilon}_{k-1}\;\middle|\;\cF^{N,\epsilon}_{k-1}\right]\\
&\qquad \qquad -  \frac{1}{2}\psi_k
                                        \E_{\QQ^N}\left[(P^N_{\tau^{N,\epsilon}_k})^2-(P^N_{\tau^{N,\epsilon}_{k-1}})^2-\sigma^2(\eta+\tau^{N,\epsilon}_k-\tau^{N,\epsilon}_{k-1})\;\middle|\;\cF^{N,\epsilon}_{k-1}\right]\Bigg]\\
&\qquad  + \sup_{\Psi}
  \left\{-\frac{\sigma^2}{2\widehat{\delta}}\Psi^2\eta+\frac{1}{2}\Psi\sigma^2\eta\right\},\label{eq:111}
  \end{align}
where $\eta > 0$ is arbitrary and where we used that $\psi_k$ is
$\cF^{N,\epsilon}_{k-1}$-measurable. The minimum over such $\psi_k$ in
the last expectation is attained for
\begin{align}
   \psi_{k}^* = \frac{\widehat{\delta}\E_{\QQ^N}\left[(P^N_{\tau^{N,\epsilon}_k})^2-(P^N_{\tau^{N,\epsilon}_{k-1}})^2-\sigma^2(\eta+\tau^{N,\epsilon}_k-\tau^{N,\epsilon}_{k-1})\;\middle|\;\cF^{N,\epsilon}_{k-1}\right]}{2\sigma^2\E_{\QQ^N}\left[\eta+\tau^{N,\epsilon}_k-\tau^{N,\epsilon}_{k-1}\;\middle|\;\cF^{N,\epsilon}_{k-1}\right]},
\end{align}
which is uniformly bounded in $N$ for $\eta>0$ due to the uniform bound on
$P^N$ up to time $\tau^N_K$. In particular, $|\psi^*_k| \leq \log N$
for sufficiently large $N$. The corresponding minimum is
\begin{align}
  &-\frac{\widehat{\delta}}{8 \sigma^2}
  \E_{\QQ^N}\left[\frac{\E_{\QQ^N}\left[(P^N_{\tau^{N,\epsilon}_k})^2-(P^N_{\tau^{N,\epsilon}_{k-1}})^2-\sigma^2(\eta+\tau^{N,\epsilon}_k-\tau^{N,\epsilon}_{k-1})\;\middle|\;\cF^{N,\epsilon}_{k-1}\right]^2}{\E_{\QQ^N}\left[\eta+\tau^{N,\epsilon}_k-\tau^{N,\epsilon}_{k-1}\;\middle|\;\cF^{N,\epsilon}_{k-1}\right]}\right]\\
&= - \frac{\widehat{\delta}}{8 \sigma^2}
  \E_{\QQ^N}\left[\left(\frac{\E_{\QQ^N}\left[(P^N_{\tau^{N,\epsilon}_k})^2-(P^N_{\tau^{N,\epsilon}_{k-1}})^2\;\middle|\;\cF^{N,\epsilon}_{k-1}\right]}{\E_{\QQ^N}\left[\eta+\tau^{N,\epsilon}_k-\tau^{N,\epsilon}_{k-1}\;\middle|\;\cF^{N,\epsilon}_{k-1}\right]}-\sigma^2\right)^2(\eta+\tau^{N,\epsilon}_k-\tau^{N,\epsilon}_{k-1})\right].
\end{align}
 Now, we just need to combine these contribution to $\cE(\QQ^N,\psi^*)$ with
 estimate~\eqref{eq:110} and the fact that the supremum in
 \eqref{eq:111} is $\eta \sigma^2 \widehat{\delta}/8$ to derive the
 claimed estimate~\eqref{eq:5new}.

For the remaining estimate~\eqref{eq:401}, consider $\psi \equiv 0$ in
the estimate~\eqref{eq:110} for $\widehat{\pi}^N(G)$. Since by absence
of arbitrage at the same time $\widehat{\pi}^N(G) \geq 0$, we can
conclude, at least for large enough $N$,
\begin{align}
  0 \leq \eta +
  \E_{\QQ^N}\left[G-\log N \sum_{k=1}^K\left|\E_{\QQ^N}\left[P^N_{\tau^{N,\epsilon}_k}-P^N_{\tau^{N,\epsilon}_{k-1}}\middle|\cF^{N,\epsilon}_{k-1}\right]\right|\right],
\end{align}
which gives~\eqref{eq:401}.
\end{proof}

The claims $H^{N,\epsilon,K}/(1-\lambda)$ of~\eqref{eq:7} are of the form required
for the previous lemma. This yields an upper bound for
super-replication prices which, however, still depends on $N$ and only
involves a process which is almost a martingale along the times of its
jumps. To get a more convenient upper bound, it will be useful to
consider processes on a slightly expanded time horizon, namely on
$[0,1+\lambda]$ rather than $[0,1]$. The payoff function $h$ can be
rescaled to $h_{1+\lambda}:D[0,1+\lambda] \to \RR$ simply by letting,
for $p \in D[0,1+\lambda]$,
\begin{align}
h_{1+\lambda}(p) \set h\left( [0,1] \ni t \mapsto p\left(t(1+\lambda)\right)\right).
\end{align}
Now, consider, for $\epsilon, \lambda>0$ fixed and
$K=K(\epsilon,\lambda)$ as in Lemma~\ref{lem:epsilondiscretization},
the class $\cD^{\epsilon,\lambda}$ of measurable processes $D$ on some
probability space $(\Omega^D,\cF^D,\P^D)$ of the form
\begin{align}\label{eq:602}
    D_t = \sum_{k=1}^K D_{{\theta}_{k-1}} 1_{[\theta_{k-1},\theta_k)}(t)+(D_{{\theta}_{K}}+\sigma W_{t-\theta_K})1_{[\theta_K,1+\lambda]}(t)
\end{align}
such that, for $k=1,\dots,K$, we have
\begin{align} \label{eq:603}
    D_0&=p_0, \quad |D_{{\theta}_{k}}-D_{{\theta}_{k-1}}| \leq 2\epsilon, \\
    \theta_0&=0, \quad
    \frac{\lambda}{K} \leq \theta_k-\theta_{k-1} \leq \frac{\lambda}{K}+\epsilon^2, \label{eq:604}
\end{align}
and
\begin{align}
    D_{\theta_{k-1}}=\E^D\left[{D_{\theta_k}} \middle|\cF^D_{\theta_{k-1}}\right]\label{eq:606}
\end{align}
where $(\cF^D_{t})_{0 \leq t \leq 1+\lambda}$ denotes the filtration
generated by $D$ and where $W$ is a Brownian motion independent of
$\cF^D_{\theta_K}$ under $\P^D$. It will also be convenient to
associate with each such $D$ the process
\begin{align*}
    \zeta^D_t \set
    \sum_{k=1}^K
     \frac{\E_{\P^D}\left[
      \left(D_{{\theta}_{k}}\right)^2
      -\left(D_{{\theta}_{k-1}}\right)^2
      \middle|\cF^{D}_{\theta_{k-1}}
      \right]}{\E_{\P^D}\left[\theta_{k}-\theta_{k-1}\middle|\cF^D_{\theta_{k-1}}\right]}1_{[\theta_{k-1},\theta_k)}(t)
      +\sigma^21_{[\theta_K,1+\lambda]}(t),
\end{align*}
which, for later use, we observe is bounded for any fixed $\lambda>0$,
$\epsilon<1/\lambda$. Indeed,
combining~\eqref{eq:606}, \eqref{eq:603}
and~\eqref{eq:604} with $K=[(c/\epsilon \lambda)^2]+1$ implies
\begin{align}
  \label{eq:623}
   |\zeta^D_t| \leq \frac{4\epsilon^2}{\lambda/K} \vee \sigma^2 \leq
  \frac{4(c+1)}{\lambda^3} \vee \sigma^2.
\end{align}
With this notation, we get the following duality estimate:
\begin{Lemma}\label{lem:uniformdualityestimate}
 For any $\epsilon,\lambda>0$ and with $K=K(\epsilon,\lambda)$ as in Lemma~\ref{lem:epsilondiscretization}, we have
 \begin{align}
     \widehat{\pi}^N&(H^{N,\epsilon,K}/(1-\lambda))\label{eq:epsilonduality}\\\nonumber
     &\leq
       \left(2+\frac{1}{4}\sigma^2\widehat{\delta}\right)\frac{\lambda}{K}
+ \sup_{D \in \cD^{\epsilon,\lambda}}\E_{\P^D}\left[\frac{h_{1+\lambda}(D)}{1-\lambda}-\frac{\widehat{\delta}}{8\sigma^2}\int_0^{1+\lambda} \left(
      \zeta^{D}_t-\sigma^2
     \right)^2\,dt\right]
 \end{align}
 for sufficiently large $N$.
\end{Lemma}
\begin{proof}
  Fix $\epsilon, \lambda>0$ and let $K \set K(\epsilon,\lambda)$ and
  $\eta \set \lambda/K$. We will use the notation from
  Lemma~\ref{lem:dualityestimate} and assume henceforth that $N$ is
  large enough for this lemma's assertion to hold true. With $G \set H^{N,\epsilon,K}$ and
  $\tau_k \set \tau^{N,\epsilon}_k$, $k=0,\dots,K$, we furthermore
  define, again for $k=0,\dots,K$,
  \begin{align}
    \label{eq:600}
     \theta^N_k &\set \tau^{N, \epsilon}_k+k\eta, \\
     D^N_{\theta^N_k} &\set p_0 + \sum_{j=1}^k
                        \left(P^N_{\tau^{N,\epsilon}_j}-\E_{\QQ^N}\left[P^N_{\tau^{N,\epsilon}_{j}}\;\middle|\;
    \cF^{D^N}_{\theta^{N}_{j-1}}\right]\right),
  \end{align}
  to specify via~\eqref{eq:602} a process
  $D=D^N$ along with a probability $\P^D \set
  \QQ^N$ that is contained in $\cD^{\epsilon,\lambda}$. Indeed, the
  martingale-like property~\eqref{eq:606} is immediate as are the
  constraints on the intervention times~\eqref{eq:604}. The increment
  restriction~\eqref{eq:603} holds since $P^N_{\tau^{N,\epsilon}_{k}}$
  is within $\epsilon$ of the
  $\cF^{D^N}_{\theta^{N}_{k-1}}$-measurable quantity
  $P^N_{\tau^{N,\epsilon}_{k-1}}$ due to the definition
  $\tau^{N,\epsilon}_k$. Moreover, it is easy to check that
  \begin{align}\label{eq:610}
   \max_{k=0,\dots,K}|P^{N,\epsilon}_{\tau^{N,\epsilon}_{k}}-D^N_{\theta^N_k}|
  \leq
    \sum_{k=1}^K\left|\E_{\QQ^N}\left[P^N_{\tau^{N,\epsilon}_k}-P^N_{\tau^{N,\epsilon}_{k-1}}\middle|\cF^{N,\epsilon}_{k-1}\right]\right|
\end{align}
  and
  \begin{align}
    \label{eq:611}
    \max_{k=0,\dots,K}\left|\frac{\tau^{N,\epsilon}_k+\lambda
    k/K}{1+\lambda}-\tau^{N,\epsilon}_k\right| \leq \lambda.
  \end{align}
  Therefore, on the event $\{\tau^{N,\epsilon}_K=1\}$ we can estimate the Skorohod-distance
  \begin{align}
    \label{eq:612}
    d(P^{N,\epsilon},D^N_{(1+\lambda)\cdot .}) \leq \lambda + \sum_{k=1}^K\left|\E_{\QQ^N}\left[P^N_{\tau^{N,\epsilon}_k}-P^N_{\tau^{N,\epsilon}_{k-1}}\middle|\cF^{N,\epsilon}_{k-1}\right]\right|
  \end{align}
 so that by Assumption~\ref{asp:hregular} on $h$ and ~\eqref{eq:401} we get
 \begin{align}
   \label{eq:613}
    \E_{\QQ^N}[H^{N,\epsilon,K}] \leq &
                                        \E_{\QQ^N}[h_{1+\lambda}(D^N)]
   \\&+ L \E_{\QQ^N}\left[\lambda +
     \sum_{k=1}^K\left|\E_{\QQ^N}\left[P^N_{\tau^{N,\epsilon}_k}-P^N_{\tau^{N,\epsilon}_{k-1}}\middle|\cF^{N,\epsilon}_{k-1}\right]\right|\right]\\
 \leq &\E_{\QQ^N}[h_{1+\lambda}(D^N)]+ L(\lambda+\|H^{N,\epsilon,K}\|_{\infty}+\eta)/\log N\\
 \leq &\E_{\QQ^N}[h_{1+\lambda}(D^N)]+ O(\frac{1}{\log N}),
 \end{align}
 where the latter estimate is due to the simple bound
$$
H^{N,\epsilon,K} \leq h(p_0)+Ld(P^{N,\epsilon},p_0) \leq h(p_0)+L
K\epsilon \text{ on } \{H^{N,\epsilon,K}>0\}.
$$
Next, from ~\eqref{eq:401}, ~\eqref{eq:610} and the fact that
 $D^N,P^N$ are uniformly bounded in $N$ we obtain
   \begin{align}
   \label{eq:614}
   &\E_{\QQ^N}\left(\int_0^{1+\lambda}(\zeta^{D^N}_t-\sigma^2)^2\,dt\right) \\
=&\E_{\QQ^N}\left(\sum_{k=1}^K\left(\frac{\E_{\QQ^N}\left[(D^N_{\tau^{N,\epsilon}_k})^2-(D^N_{\tau^{N,\epsilon}_{k-1}})^2\;\middle|\;\cF^{N,\epsilon}_{k-1}\right]}
{\E_{\QQ^N}\left[\eta+\tau^{N,\epsilon}_k-\tau^{N,\epsilon}_{k-1}\;\middle|\;\cF^{N,\epsilon}_{k-1}\right]}-\sigma^2\right)^2
(\eta+\tau^{N,\epsilon}_k-\tau^{N,\epsilon}_{k-1})\right)\\
\leq&\E_{\QQ^N}\left(\sum_{k=1}^K\left(\frac{\E_{\QQ^N}\left[(P^N_{\tau^{N,\epsilon}_k})^2-(P^N_{\tau^{N,\epsilon}_{k-1}})^2\;\middle|\;\cF^{N,\epsilon}_{k-1}\right]}
{\E_{\QQ^N}\left[\eta+\tau^{N,\epsilon}_k-\tau^{N,\epsilon}_{k-1}\;\middle|\;\cF^{N,\epsilon}_{k-1}\right]}-\sigma^2\right)^2
(\eta+\tau^{N,\epsilon}_k-\tau^{N,\epsilon}_{k-1})\right)+O(\frac{1}{\log N}).
 \end{align}
 Using~\eqref{eq:613} and~\eqref{eq:614} in the estimate~\eqref{eq:5new}
 provided by Lemma~\ref{lem:dualityestimate}
 it thus follows that $\widehat{\pi}^N(H^{N,\epsilon,K}/(1-\lambda))$ cannot be
 larger than the right-hand side of~\eqref{eq:epsilonduality} for
 sufficiently large $N$.
\end{proof}

Letting $\epsilon \downarrow 0$ in the above expression will be made
possible by the following tightness result:

\begin{Lemma}\label{lem:sequentialcompactness}
  For $\lambda>0$ fixed, any sequence $D^m \in \cD^{1/m,\lambda}$,
  $m=1,2,\dots$, contains a subsequence along which
  $\Law(D^m,\int_0^. \zeta^{D^m}_s \,ds \;|\;\P^{D^m})$ converges
  weakly on $D[0,1+\lambda]$ to $\Law(M,\langle M \rangle \;|\;\P^M)$
  for some continuous martingale $M=(M_t)_{0 \leq t \leq 1+\lambda}$
  on a suitable probability space $(\Omega^M,\cF^M,\P^M)$.
\end{Lemma}
\begin{proof}
  Let $K^m \set K(1/m,\lambda)$ as in
  Lemma~\ref{lem:epsilondiscretization} and denote by
  $(\theta^m_k)_{k=0,\dots,K^m}$ the times associated with $D^m$
  via~\eqref{eq:602}--\eqref{eq:606}; let furthermore
  $\P^m \set \P^{D^m}$ denote the associated probability. For any
  $m=1,2,\dots$, we denote by $\{\hat D^m_t\}_{t=0}^{1+\lambda}$ the
  continuous linear interpolation of $D^m$; observe that after time
  $\theta^m_K$, this amounts to $\hat D^m_t \set D^m_{\theta^m_K}+\sigma W_{t-\theta^m_K}$ where $W$ is a Brownian motion  independent of $\mathcal F^{D^m}_{\theta^m_K}$.

  We will verify the Kolmogorov tightness criterion for these
  processes $\hat D^m$, $m=1,2,\dots$. So, take $m \in \{1,2,\dots\}$
  and fix $0\leq t_1<t_2\leq 1+\lambda$. Define the random times
$$\eta_i=K^m\wedge\min\left\{k \in \{0, \dots, K^m\}\;:\;\theta^{m}_k\geq t_i\right\}, \ \ i=1,2.$$
The discrete-time process $\{D^{m}_{\theta^{m}_{k}}\}_{k=0,\dots,K_m}$
is a martingale with respect to the filtration generated by
$(\theta^m_k,D^m_{\theta^m_k})_{k=0,\dots,K^m}$ and
$\eta_1,\eta_2$ are stopping times with respect to this filtration.
 Thus, from the Burkholder--Davis--Gundy inequality,
\begin{align}
  \E_{\mathbb P^{m}}\left[\left| D^{m}_{\theta^{m}_{\eta_2}}- D^{m}_{\theta^{m}_{\eta_1}}\right|^4\right]&
   \leq O(1)
\E_{\mathbb P^{m}}\left[\left(\sum_{j=\eta_1+1,\dots,\eta_2}|D^{m}_{\theta^{m}_{j}}-
D^{m}_{\theta^{m}_{j-1}}|^2\right)^2\right]\\ \label{eq:6.1}
&
\leq O(m^{-4})\E_{\mathbb P^{m}}\left[ |\eta_2-\eta_1|^2\right]
 \end{align}
where the last inequality follows from~\eqref{eq:603} which ensures that
the jumps of $D^m$ are bounded by $2/m$.

Next, since the time between two subsequent jumps is at least
$\lambda/K^m=O(1/m^2)$ and the jumps of $D^m$ are bounded by $2/m$, we
obtain that the size of the (random) slope for the linear
interpolation process $\hat D^{m}$ is at most of order $O(m)$.  This
together with the fact that the time between two subsequent jumps is
less than or equal to $\lambda/K^m+1/m^2$ yields
\begin{align}
|&\hat D^{m}_{t_2}-\hat D^{m}_{t_1}|\\
&\leq \left|\hat
D^{m}_{t_2\wedge\theta^{m}_{K}}-
\hat
D^{m}_{t_1\wedge\theta^{m}_{K}}\right|+\sigma\left|W_{t_2\vee\theta^{m}_{K}}-W_{t_1\vee\theta^{m}_{K}}\right|\\
&\leq
1_{\{t_2>t_1+1/m^2\}}\left|D^{m}_{\theta^{m}_{\eta_2}}-D^{m}_{\theta^{m}_{\eta_1}}\right|\nonumber\\
&+2 O(m)  \left((t_2-t_1)\wedge\left(\lambda/K_m +1/m^2\right)\right)+
\sigma\left|W_{t_2\vee\theta^{m}_{K}}-W_{t_1\vee\theta^{m}_{K}}\right|.\label{eq:6.2}
\end{align}
Using again that the time between two subsequent jumps is at least
$\lambda/K_m=O(1/m^2)$, we obtain that if $t_2>t_1+1/m^2$ then
$\eta_2-\eta_1=O(m^2) (t_2-t_1)$. Thus, from
\eqref{eq:6.1}--\eqref{eq:6.2} and the elementary inequalities
\begin{eqnarray*}
&(t_2-t_1)\wedge\left(\lambda/K^m +1/m^2\right)\leq  O(1/m)\sqrt{t_2-t_1} ,\\
&(z_1+z_2+z_3)^4\leq 81(z^4_1+z^4_2+z^4_3),
\end{eqnarray*}
we obtain
$\E_{\P^m}\left[|\hat D^m_{t_2}-\hat D^{m}_{t_1}|^4\right]=O((t_2-t_1)^2)
$ and tightness follows.

From Prokhorov's theorem we conclude that there exists a
subsequence (still denoted by $m$) and a continuous process
$M=(M_t)_{0 \leq t \leq 1+\lambda}$ that converges in law to some
continuous process $M$ on a suitable probability space $(\Omega^M,\cF^M,\P^M)$.
The obvious inequality $\sup_{0\leq t\leq 1+\lambda} |\hat D^{m}_t- D^{m}_t|\leq 2/m$ yields
the same convergence also for $D^{m}$.

Let us argue next that $M$ is a martingale with respect to its own
filtration. Fix $m$, let $(\cF^m_t)_{0 \leq t \leq 1+\lambda}$ be
the usual (right continuous and complete) filtration generated by
$D^m$ and consider the (RCLL) martingale
$$\tilde D^{m}_t=\E_{\P^{m}}[D^m_{1+\lambda}|\cF^m_t], \quad 0
\leq t \leq 1+\lambda.$$
Recall that the time between
two subsequent jumps is bounded from below. Hence,
\begin{align}
\tilde D^{m}_{\theta^{m}_{k}}&=\E_{\P^{m}}[D^{m}_{\theta^m_{K_m}}|\cF_{\theta^m_{k}}]
=D^{m}_{\theta^{m}_{k}},  \ k=0,1,...,K^m,    \label{eq:neww}
\end{align}
where the last equality follows from~\eqref{eq:606}.
From~\eqref{eq:neww} and the estimate
$$\max_{k=1, \dots, K^m}|D^{m}_{\theta^{m}_{k}}-D^{m}_{\theta^{m}_{k-1}}|\leq 2/m$$
we get $\|\tilde D^m-D^m\|_{\infty}\leq 4/m$.
Thus, the martingales $\tilde D^{m}$, $m=1,2,\dots$, are uniformly integrable and converge weakly to
$M$. From Theorem~5.3 in \cite{Whitt:07} we conclude that $M$ is a (continuous) martingale.

Now, we prove that
$\left(D^{m},\int_{0}^{\cdot}\zeta^{D^m}_s ds\right)$ converges in law
to $(M,\langle  M\rangle)$. For any $m=1,2,\dots$, let the
quadratic variation of the martingale $\tilde D^{m}$ be denoted by $([\tilde
D^{m}]_t)_{0 \leq t \leq 1+\lambda}$.
Theorem 5.5 in~\cite{Whitt:07} then yields the converge in law of
$(\tilde D^{m}, [\tilde D^{m}])$ to $(M,\langle M\rangle)$.
Thus, in order to
complete the proof, it sufficient to establish that
\begin{equation}\label{eq:6.probab}
\lim_{m\rightarrow\infty}\E_{\P^{m}}\left[\sup_{0\leq t\leq 1+\lambda}
\left|\int_{0}^ t\zeta^{D^m}_s ds-[\tilde D^{m}]_t\right|\right]=0.
\end{equation}
To that end, note that by the Burkholder--Davis--Gundy inequality
\begin{align}
\E_{\P^{m}}&\left[\max_{k=1, \dots, K^m}\left([\tilde D^{m}]_{\theta^{m}_{k}}-
[\tilde D^{m}]_{\theta^{m}_{k-1}}\right)^2\right]
\\&\leq \sum_{k=1}^{K^m}\E_{\P^{m}}\left[\left([\tilde D^{m}]_{\theta^{m}_{k}}-
[\tilde D^{m}]_{\theta^{m}_{k-1}}\right)^2\right]\nonumber\\
&\leq O(1) \sum_{k=1}^{K^m}\E_{\P^{m}}\left[\sup_{\theta^{m}_{k-1}\leq t\leq\theta^{m}_{k}}
 \left|\tilde D^{m}_t-
\tilde D^{m}_{\theta^{m}_{k-1}}\right|^4\right]\nonumber\\
&\leq O(1) K^m O(1/m^4)=O(1/m^2).\label{eq:6.4}
\end{align}
Next, observe that $([\tilde
D^{m}]_{\theta^{m}_{k}}-\int_{0}^{\theta^{m}_{k}}\zeta^{D^m}_s
ds)_{k=0, \dots, K_m}$
is a martingale. Hence, applying first the Doob--Kolmogorov inequality and Ito's
isometry and finally also \eqref{eq:604}, \eqref{eq:623}, we conclude
\eqref{eq:6.4}
\begin{align}
\E_{\P^{m}}&\left[\max_{k=0, \dots,\leq K^m}
\left([\tilde D^{m}]_{\theta^{m}_{k}}-\int_{0}^{\theta^{m}_{k}}\zeta^{D^m}_s ds\right)^2\right]\\
&\leq 4\E_{\P^{m}}\left[\sum_{k=1}^{K^m} \left([\tilde D^{m}]_{\theta^{m}_{k}}-[\tilde D^{m}]_{\theta^{m}_{k-1}}+\int_{\theta^m_{k-1}}^{\theta^m_k}|\zeta^{D^m}_s| ds\right)^2\right]\\
&\leq 8\E_{\P^{m}}\left[\sum_{k=1}^{K^m} \left([\tilde D^{m}]_{\theta^{m}_{k}}-[\tilde D^{m}]_{\theta^{m}_{k-1}}\right)^2\right]\\
&+8 K^m \|\zeta^{D^m}\|^2_{\infty} (\lambda/K^m+1/m^2)^2=O(1/m^2).\label{eq:6.3}
 \end{align}
Finally, by combining \eqref{eq:604}, \eqref{eq:623} and applying the Jensen inequality for \eqref{eq:6.4}--\eqref{eq:6.3} we get
\begin{align}
\E_{\P^{m}}&\left[\sup_{0\leq t\leq 1+\lambda}\left|\int_{0}^t\zeta^{D^m}_s ds-[\tilde D^{m}]_t\right|\right]\\
=&\E_{\P^{m}}\left[\sup_{0\leq t\leq \theta^m_{K^m}}\left|\int_{0}^t\zeta^{D^m}_s ds-[\tilde D^{m}]_t\right|\right]\\
\leq&\|\zeta^{D^m}\|_\infty (\lambda/K^m+1/m^2)\\
&+\E_{\P^{m}}\left[\max_{k=1, \dots, K^m}\left|[\tilde D^{m}]_{\theta^{m}_{k}}-
[\tilde D^{m}]_{\theta^{m}_{k-1}}\right|\right]\\
&+\E_{\P^{m}}\left[\max_{k=0,\dots, K^m}\left|\int_{0}^{\theta^m_k}\zeta^{D^m}_s ds-
[\tilde D^{m}]_{\theta^{m}_{k}}\right|\right]=O(1/m)
\end{align}
and \eqref{eq:6.probab} follows.
\end{proof}
We will need the following stability result.
\begin{Lemma}\label{lem:continuity}
  Let $(\Omega,\cF,(\cF_t),\P)$ be an arbitrary filtered probability
  space and suppose $h$ satisfies Assumption~\ref{asp:hregular}. Then the function
\begin{equation} \label{eq:definitio}
 F(z)=\sup_{M}\E\left[z_1h_{}(M)-z_2\int_0^{1} \left(
      \frac{d\langle M \rangle_t}{dt}-z_3\right)^2dt\right], \quad
  z=(z_1,z_2,z_3)\in (0,\infty)^3,
\end{equation}
where the supremum is taken over all continuous martingales
$M=(M_t)_{0 \leq t \leq 1}$ starting in $M_0=p_0$ that have a quadratic
variation $\langle M \rangle$ which is absolutely continuous with bounded
density $\frac{d\langle M \rangle_t}{dt}$.
\end{Lemma}
\begin{proof}
It is sufficient to prove the statement for the function
\begin{align}\nonumber
\hat F(z) &\set F(z)+z_2z^2_3\\&=\sup_{M}\E\left[z_1h_{}(M)-z_2\int_0^{1} \left(
      \frac{d\langle  M \rangle_t}{dt}\right)^2dt+2 z_2z_3\langle  M \rangle_1 \right].
      \label{eq:definition}
      \end{align}
      From the Doob--Kolmogorov inequality, the Jensen inequality and
      the estimate $h(p) \leq \|p-p_0\|^2_{\infty}+c$ for $c=c(1)$ (as
      defined in Lemma~\ref{lem:epsilondiscretization}) we obtain
  \begin{align}\nonumber
  \E&\left[z_1h_{}(M)-z_2\int_0^{1} \left(
      \frac{d\langle  M \rangle_t}{dt}\right)^2dt+2z_2 z_3\langle  M \rangle_1 \right]\\\nonumber
      &\leq z_1 c+(4z_1+2 z_2 z_3)\E\langle  M \rangle_1-z_2\E \left[\int_0^{1} \left(
      \frac{d\langle  M \rangle_t}{dt}\right)^2 dt\right]\\\nonumber
      &\leq z_1 c+(4z_1+2 z_2 z_3)\sqrt{\E \left[\int_0^{1} \left(
      \frac{d\langle  M \rangle_t}{dt}\right)^2dt\right]}-z_2 \E \left[\int_0^{1} \left(
      \frac{d\langle  M \rangle_t}{dt}\right)^2dt\right].\\
      \label{eq:estimate}
      \end{align}
      On the other hand by taking $M\equiv p_0$ we obtain
      $\hat F\geq 0$ (recall that $h$ is nonnegative). Hence, on the
      right hand side of~\eqref{eq:definition} we can restrict the
      supremum to the set of martingales for which the right hand size
      of~\eqref{eq:estimate} is nonnegative.

We conclude that for a bounded set $O\subset\mathbb R^3$ with
$\inf_{z\in O}z_2>0$ there exists $\Theta=\Theta(O)$ such that for any
$z\in O$ we have
\begin{equation}
 \hat F(z)=\sup_{M}\E\left[z_1h_{}(M)-z_2\int_0^{1} \left(
      \frac{d\langle  M \rangle_t}{dt}\right)^2dt+2 z_2z_3\langle  M \rangle_1 \right]
      \end{equation}
      where the supremum is taken over the class $\cM_\Theta$ of continuous martingales as
      in the formulation of this lemma which satisfy in addition that
      $$\E\left[\int_0^{1} \left( \frac{d\langle M
            \rangle_t}{dt}\right)^2dt\right]\leq \Theta.$$  In
      particular, we obtain that $\hat F(z)<\infty$.  By applying
      again the Doob--Kolmogorov inequality and the Jensen inequality
      it follows that for any $z,\tilde z\in O$ we have
      \begin{align}
  |\hat F(z)&-\hat F(\tilde z)|\\
  &\leq \sup_{M \in \cM_\Theta}\left(\E\left[z_1h_{}(M)-z_2\int_0^{1} \left(
      \frac{d\langle  M \rangle_t}{dt}\right)^2dt+2z_2 z_3\langle  M \rangle_1 \right]\right.\\
      &\qquad\qquad\left.-\E\left[\tilde z_1 h_{}(M)-\tilde z_2\int_0^{1} \left(
      \frac{d\langle  M \rangle_t}{dt}\right)^2dt+2\tilde z_2 \tilde z_3\langle  M \rangle_1 \right]\right)\\
  &\leq |z_1-\tilde z_1| (c+4\sqrt{\Theta})
      +|z_2-\tilde z_2|\Theta+2|z_2z_3-\tilde z_2\tilde z_3|\sqrt\Theta
  \end{align}
and continuity follows.
\end{proof}

We now have all the pieces in place that we need for the

\paragraph{Completion of the proof of the upper bound.}

Fix $\lambda>0$. For $m=1,2,\dots$, choose $D^m \in \cD^{1/m,\lambda}$
that get within $1/m$ of the supremum in~\eqref{eq:epsilonduality} for
$\epsilon=1/m$. For this sequence, let $M^\lambda$ be a continuous
martingale on $(\Omega^\lambda,\cF^\lambda,\P^\lambda)$ as in
Lemma~\ref{lem:sequentialcompactness}. By Skorohod's representation
theorem, we can find copies of $D^m$, $m=1,2,\dots,$ and $M^\lambda$
(which to alleviate notation we denote by the same symbols) with the
same respective distributions but specified jointly on a suitable probability space
$(\Omega,\cF,\P)$ such that almost surely
$(D^m,\int_0^. \zeta^{D^m}_s ds)$ converges uniformly to
$(M^\lambda,\langle M^\lambda \rangle)$ as $m \uparrow \infty$. From~\eqref{eq:603} we have
$M^\lambda_0=p_0$ and from~\eqref{eq:623} the quadratic
variation $\langle M^\lambda \rangle$ is absolutely continuous and the volatility process
$\frac{d\langle M^\lambda\rangle_t}{dt}$ is bounded.

Now, use the regularity of $h$ as imposed by Assumption~\ref{asp:hregular} to conclude
\begin{align}\label{eq:9}
    \lim_m \E\left[{h_{1+\lambda}(D^m)}\right] = \E\left[{h_{1+\lambda}(M^\lambda)}\right]
\end{align}
by dominated convergence. Moreover, we can estimate
\begin{align}
    \liminf_m \E\left[\int_0^{1+\lambda} \left(
      \zeta^{D^m}_t-\sigma^2
     \right)^2\,dt\right]
     &\geq \E\left[\int_0^{1+\lambda} \left(
      \frac{d\langle M^\lambda \rangle_t}{dt}-\sigma^2
     \right)^2\,dt\right].\label{eq:10new}
\end{align}
Indeed, observing that the $\zeta^{D^m}$ are uniformly bounded for
$m>1/\lambda$ (cf.~\eqref{eq:623}), we can apply Lemma~A1.1 in
\cite{DelbSch:94} to get
$\tilde{\zeta}^m\in
\textrm{conv}(\zeta^{D^m},\zeta^{D^{m+1}},\dots)$, $m=1,2,\dots$,
converging $\P \otimes dt$-almost everywhere to some process
$\zeta$. In fact, $\zeta=d\langle M^\lambda \rangle/dt$ because by dominated
convergence
$\int_0^. \zeta_t dt = \lim_m \int_0^. \tilde{\zeta}^m_t dt=\lim_m
\int_0^. \zeta^{D^m}_t dt = \langle M^\lambda \rangle$. As a consequence, the
estimate~\eqref{eq:10new} holds by the convexity of
$\zeta \mapsto \E\left[\int_0^{1+\lambda} \left( \zeta_t-\sigma^2
  \right)^2\,dt\right]$ and Fatou's lemma.

As $K(1/m,\lambda) \to \infty$ for $m \uparrow \infty
$, it now follows from Lemma~\ref{lem:uniformdualityestimate} that, for any $\lambda>0$,
\begin{align}
    \limsup_m &\limsup_N \widehat{\pi}^N(H^{N,1/m,K(1/m,\lambda)}/(1-\lambda))\nonumber\\
    &\leq
      \E\left[\frac{h_{1+\lambda}(M^\lambda)}{1-\lambda}-\frac{\widehat{\delta}}{8
      \sigma^2}\int_0^{1+\lambda} \left(
      \frac{d\langle M^\lambda \rangle_t}{dt}-\sigma^2
     \right)^2\,dt\right] \\
     &=\E\left[\frac{h_{}(\tilde M^\lambda)}{1-\lambda}-\frac{\widehat{\delta}}{8
      \sigma^2(1+\lambda)}\int_0^{1} \left(
      \frac{d\langle \tilde M^\lambda \rangle_t}{dt}-\sigma^2(1+\lambda)
     \right)^2\,dt\right]\label{eq:13}
\end{align}
where $\tilde M^\lambda$ is the martingale given by
$\tilde M^\lambda_t=M^{\lambda}_{(1+\lambda)t}$, $0 \leq t\leq 1$.

By applying Lemma~\ref{lem:continuity} we see that for
$\lambda \downarrow 0$ the expectation in~\eqref{eq:13} cannot be
larger than
$\sup_{M} \E\left[h(M)-\frac{\widehat{\delta}}{8\sigma^2}\int_0^{1}
  \left( \frac{d\langle M \rangle_t}{dt}-\sigma^2 \right)^2\,dt\right]
$ where the supremum is taken over all the continuous martingales
$M=(M_t)_{0 \leq t \leq 1}$ considered in Lemma~\ref{lem:continuity}.
Using the randomization technique of
Lemma~7.2 in \cite{DolinskySoner:13}, we thus find that
$$\limsup_{\lambda \downarrow 0}\limsup_m \limsup_N
\widehat{\pi}^N(H^{N,1/m,K(1/m,\lambda)}/(1-\lambda))$$ is dominated
by the supremum on the right hand side of~\eqref{eq:4}. In view of the
estimate~\eqref{eq:6} from Lemma~\ref{lem:epsilondiscretization} in
conjunction with Lemma~\ref{lem:quadraticclaim}, this implies the
desired upper bound~\eqref{eq:4} for the case $x_0=0$ and $\zeta_0=0$.

For the case where $x_0 \not=0$ or $\zeta_0>0$ we need to establish that
\begin{align}
     \limsup_N \pi^N(h(P^N)) &\leq \sup_{\nu \in \cD} \E_{\PP^W}\left[h(P^\nu)-\frac{r \delta}{8 \sigma^2(2-r)}\int_0^1 |\nu_t^2-\sigma^2|^2\,dt\right]\\ &
                                                                   \qquad - P_0x_0-\frac{\iota}{2}x_0^2.
 \label{eq:1000new}
 \end{align}
For the $N$--step market we use the first $N^{1/3}$ steps to liquidate with a constant speed the initial number of shares $x_0$.
The result is that after $N^{1/3}$  steps, the portfolio value
will be $P_0x_0+\frac{\iota}{2}x_0^2+O(N^{-1/6})$ and the spread will be bounded
by $\zeta_0+\frac{x_0}{\delta}$. The number of shares is zero.

In the next $N^{1/3}$ steps we do not trade at all, and so the spread will become of order
$O\left((1-r)^{N^{1/3}}\right)$. Observe that for any $\tilde\delta>\delta$, we have that for sufficiently large $N$,
$$\delta \left(z+O\left((1-r)^{N^{1/3}}\right)\right)^2 \leq
\tilde\delta z^2+(1-r)^{N^{1/4}} \text{ for all } z \geq 0.$$ From
Lemma~\ref{lem:wealthdynamics} we conclude that the limsup of the
original prices $\pi^N(h(P^N))$ is less than or equal to the limsup of
the superhedging prices which correspond to the market depth
$\tilde\delta>\delta$, the same resilience $r$ and an initial position
$\tilde x_0=\tilde \zeta_0=0$ minus
$P_0x_0+\frac{\iota}{2}x_0^2$. Thus, by taking
$\tilde\delta\downarrow \delta$ and applying~\eqref{eq:4} (for
$\tilde\delta$ instead of $\delta$) and by using Lemma
\ref{lem:continuity}, we obtain~\eqref{eq:1000new}.  Let us notice
that we should apply~\eqref{eq:4} for a shift in time of the original
price process $P^N$. Since the shift in time is of order $O(N^{1/3})$
and $h$ is Lipschitz continuous, the difference between the original
payoff $h(P^N)$ and the modified one is of order
$O(N^{1/3}N^{-1/2})=O(N^{-1/6})$ which is vanishing in the limit
$N\rightarrow\infty$.

\section*{Acknowledgments}
The author YD was partially
supported by the ISF grant no 160/17.

\bibliographystyle{spbasic}

\end{document}